\documentclass{article}

\usepackage{amsfonts}
\usepackage{amsmath}
\usepackage{graphicx}
\usepackage{epstopdf}
\usepackage{braket}
\usepackage{Preamble}

\usepackage{enumitem}

\renewcommand{\supp}{\mathrm{supp}\,}

\newcommand{\ind}{\mathrm{Ind}}

\title{Absolutely Continuous Edge Spectrum of Hall Insulators on the Lattice}

\author{Alex Bols and Albert H. Werner \\
\normalsize
QMATH, Department of Mathematical Sciences, University of Copenhagen \\
\normalsize
Universitetsparken 5, 2100 Copenhagen, Denmark\\
\normalsize
e-mail : \texttt{alex-b@math.ku.dk, werner@math.ku.dk}\\
}

\date{}

\begin{document}

\maketitle

\abstract{The presence of chiral modes on the edges of quantum Hall samples is essential to our understanding of the quantum Hall effect. In particular, these edge modes should support ballistic transport and therefore, in a single particle picture, be supported in the absolutely continuous spectrum of the single-particle Hamiltonian. We show in this note that if a free fermion system on the two-dimensional lattice is gapped in the bulk, and has a nonvanishing Hall conductance, then the same system put on a half-space geometry supports edge modes whose spectrum fills the entire bulk gap and is absolutely continuous.}

\section{Introduction}

In a finite Hall insulator, the Hall conductance is quantized to remarkable accuracy in multiples of $e^2 / h$, where $e$ is the charge of the electron, and $h$ is Planck's constant. The Hall current consists of  a bulk and an edge contribution which combine to give the quantized result. The edge current arises because the driving potential results in a skewed occupation of the \emph{edge states} of the sample. These edge states are chiral, \ie they carry a ballistic current predominantly either clockwise or counterclockwise along the edge, so a different occupation of edge modes on opposite sides of the sample leads to a net current perpendicular to the potential drop. The existence of such chiral edge modes is therefore essential to our understanding of the quantum Hall effect \cite{Laughlin1981, Halperin1982}.

The chiral edge modes must allow ballistic motion along the edge and  should therefore correspond to absolutely continuous spectrum (see for example theorem 5.7 in \cite{Teschl2009}). We show in this note that non-trivial Hall insulators modelled by free fermions on the lattice $\Z^2$ inevitably yield edge modes corresponding to absolutely continuous spectrum, filling the entire bulk gap.

The presence of edge states corresponding to (absolutely) continuous spectrum has been shown for the Landau Hamiltonian in the continuum with a weak disorder potential, where the edge is introduced either by a steep edge potential or appropriate half-plane boundary conditions \cite{MacrisMartinPule1999, FrohlichGrafWalcher2000, DeBievrePule2002, HislopSoccorsi2008, BrietHislopRaikovSoccorsi2009}. Except for the first, all these works rely on Mourre estimates to conclude purely absolutely continuous spectrum in the bulk gap. 


In this note we do not use Mourre estimates, but instead appeal to the bulk-edge correspondence for Hall insulators, which has been proven in various guises for free fermions on the lattice \cite{SchulzKellendonkRichter2000, KellendonkRichterSchulzbaldes2002, ElbauGraf2002, Macris2003preprint, ElgartGrafSchenker2005, SchulzBaldesProdan2016, ShapiroEtAl2020}. The topological index characterizing the edge system in this bulk-edge correspondence is a property of a unitary operator which is defined in terms of the edge states. Using a result from \cite{AschBourgetJoye2020} one immediately gets that this unitary has absolutely continuous spectrum if the topological index is nonzero. We then show that the Hamiltonian of the half-plane system inherits this absolutely contiuous spectrum.

A possible concrete advantage of our method is the following. The techniques using Mourre estimates lead to \emph{purely} continuous spectrum in the gap. That method therefore seems to exclude from the start the possibility of dealing with a mobility gap. Our method may be more promising in this regard.


\section{Setup and Result}

We consider free electrons moving on the lattice $\Z^2$ modelled by a local bulk Hamiltonian $H$ on $l^2(\Z^2)$. Associated to this bulk Hamiltonian there is an edge Hamiltonian $\hat H$ acting on $l^2(\Z \times \N)$ which is obtained by restricting the bulk Hamiltonian to the half-space. We spell out detailed assumptions on the bulk and edge Hamiltonians at the end of this section.

We assume that for some open interval $\Delta$
$$ \Delta \cap \sigma(H) = \emptyset. $$
Such an interval $\Delta$ is called a \emph{bulk gap}. We let $P_F = \chi_{\leq \mu}(H)$ for any $\mu \in \Delta$ denote the Fermi projection.

\begin{theorem} \label{thm:main_theorem}
	Let $H$ be a bulk Hamiltonian with bulk gap $\Delta$ and let
	$$ \ind_B(H, \Delta) := -\ind(P_F, U) \in \Z $$
	be the \emph{bulk index}, where $U := \ed^{\iu \arg \vec X}$ with $\vec X = (X_1, X_2)$ the vector of position operators.

	If $\ind_B \neq 0$ then
	$$ \Delta \subset \sigma_{ac}(\hat H). $$
	\ie the edge Hamiltonian $\hat H$ has absolutely continuous spectrum everywhere in the bulk gap $\Delta$.
\end{theorem}

The bulk index $\ind(U, P_F)$ is defined as the index of the pair of projections $U P_F U^\dag$ and $P_F$ \cite{AvronSeilerSimon1994}:
\begin{definition} \label{def:index}
	If $U$ is unitary and $P$ is a projection such that the difference
	$$ U^\dag P U - P $$
	is compact, then we define
	$$ \ind(P, U) := \dim \ker (P - U P U^\dag - \I) - \dim \ker ( P - U P U^\dag + \I) \in \Z. $$
\end{definition}
This index is stable under norm-continuous changes of both $U$ and $P$ as long as it remains well defined, see \cite{AvronSeilerSimon1994}.

\begin{remarks*}
	\begin{enumerate}
		\item It is believed that the bulk index equals the Hall conductance of the material as obtained from the Green-Kubo formula, up to a factor $e^2 / h$ with $e$ the charge of the electron, and $h$ Planck's constant. Under very natural homogeneity assumptions on the Hamiltonian (in particular, if the Hamiltonian is translation invariant) this has been proven, see \cite{BellissardEtAl1994, AizenmanGraf1998}.
		\item It is shown in \cite{AizenmanGraf1998} that the difference $P_F - U P_F U^\dag$ with $U$ as in the statement of theorem \ref{thm:main_theorem} is indeed compact, so the bulk index is well defined.
	\end{enumerate}	
\end{remarks*}

We now spell out the assumtions on the Hamiltonians $H$ and $\hat H$.

Lattice points are labelled by $\vec x = (x_1, x_2) \in \Z^2$ and we have corresponding states
$$ | \vec x \rangle : \Z^2 \rightarrow \C : \vec y \mapsto \delta_{\vec x, \vec y}. $$
These states form an orthonormal basis of the Hilbert space.

The bulk Hamiltonian $H$ is a self-adjoint operator on $l^2(\Z^2)$ which is exponentially local:
$$  \abs{\langle \vec x, \, H \,  \vec y \rangle } \leq C \ed^{-\norm{\vec x - \vec y} / \xi} $$
for some $C < \infty$, $\xi > 0$ and all $\vec x, \vec y \in \Z^2$.

The Hilbert space for the edge system is $\hat \caH = l^2(\Z \times \N)$. The natural inclusion of the half-space lattice $\Z \times \N$ in the bulk lattice $\Z \times \Z$ induces an injection
$$ \hat \caH \xrightarrow{\iota} \caH. $$

The half-space Hamiltonian $\hat H$ is a self-adjoint operator on $l^2(\Z \times \N)$ that agrees with the bulk Hamiltonian in the bulk:
$$ \abs{ \langle \vec x, \, (\iota^{\dag} H \iota - \hat H) \, \vec y \rangle } \leq C \ed^{- \norm{\vec x - \vec y}/\xi   - y_2 / \xi' } $$
for some $C < \infty, \xi, \xi' > 0$ and all $\vec x, \vec y \in \Z \times \N$. \ie up to a boundary condition which is exponentially localized near the edge of the system, the edge Hamiltonian equals the restriction of the bulk Hamiltonian to the half-space $\Z \times \N$. In particular, the half-space Hamiltonian is itself exponentially local.

\section{Edge index and bulk-edge correspondence}

Let $g : \R \rightarrow [0, 1]$ be a smooth function interpolating from $1$ to $0$ and such that $g'$ is supported in the bulk gap $\Delta$. We have
$$ P_F = g(H). $$

Consider now the unitary $W_g(\hat H)$ where $W_g$ is the function
\begin{equation} \label{def:W_g}
	W_g : \R \rightarrow \C : x \mapsto \ed^{2 \pi \iu x }.
\end{equation}
This unitary is local and supported near the edge of the half-space. In particular:
\begin{lemma}[\cite{ElbauGraf2002}] \label{lem:W_g_locality}
	Let $\hat \Pi_1$ denote the projection on $\{ \vec x \in \Z \times \N \, | \, x_1 \geq 0 \}$, then the commutator $[W_g(\hat H), \hat \Pi_1 ]$ is trace class.
\end{lemma}

This lemma follows immediately from lemmas A.2. and A.3. in \cite{ElbauGraf2002}.

It follows that $W_g(\hat H)^\dag \hat \Pi_1 W_g(\hat H) - \hat \Pi_1 = W_g(\hat H)^\dag [\hat \Pi_1, W_g(\hat H)]$ is also trace class and in particular compact so the index
$$ \ind_E(H, \Delta) := \ind(\Pi_1, W_g(\hat H)) \in \Z $$
is well defined. We call this integer the \emph{edge index}.

\begin{remark*}
	The edge index may a priori depend on the boundary conditions defining the edge Hamiltonian $\hat H$. The following theorem implies that this is not the case, justifying our notation $\ind_E(H, \Delta)$.
\end{remark*}

\begin{theorem}[Theorem 2.11 of \cite{ShapiroEtAl2020}] \label{thm:bulk-edge}
	Under the above assumptions on $H$,
	$$ \ind_B(H, \Delta) = \ind_E(H, \Delta) $$
	for any bulk gap $\Delta$.
\end{theorem}

\begin{remark*}
	In \cite{ShapiroEtAl2020} the bulk and edge indices are given als Fredholm indices. The equivalence to the definition \ref{def:index} is established in theorem 5.2. in \cite{AvronSeilerSimon1994}.
\end{remark*}









\section{Proof}

Theorem \ref{thm:main_theorem} is an almost direct consequence of the following proposition from \cite{AschBourgetJoye2020}.

\begin{proposition}[Theorem 2.1. of \cite{AschBourgetJoye2020}] \label{prop:abs_cont_spectrum}
	Let $U$ be a unitary operator on a Hilbert space and $P$ an orthogonal projection such that $[U, P]$ is trace class, then the index $\ind(U, P)$ is a well-defined finite integer. If $\ind(U, P) \neq 0$, then the absolutely continuous spectrum of $U$ is the entire unit circle.
\end{proposition}

The unitary $W_g(\hat H)$ and the projecion $\hat \Pi_1$ satisfy the assumptios of this proposition due to lemma \ref{lem:W_g_locality}.

We are now ready to give the proof of the main theorem.

\begin{proofof}[Theorem \ref{thm:main_theorem}]
	Using theorem \ref{thm:bulk-edge}, if $\ind_B(H, \Delta) \neq 0$ also $\ind_E(H, \Delta) = \ind(W_g(\hat H), \hat \Pi_1) \neq 0$. It follows then from lemma \ref{lem:W_g_locality} and proposition \ref{prop:abs_cont_spectrum} that the absolutely continuous spectrum of the edge unitary $W_g(\hat H)$ is the whole unit circle.

	For any $\ep < \abs{\Delta}/2$ let $\Delta_{\ep} = \{ x \, | \, \dist(x, \Delta^c) > \ep \}$ be the gap $\Delta$ shrunk by $\ep$. Similarly, for any $\delta < \pi$ let $\Delta'_{\delta} =  \{ \ed^{\iu \theta} \, | \, \theta \in (\delta, 2\pi -\delta) \}$ be the unit circle with a closed $\delta$-neighbourhood around $1$ excluded. We can choose $g$ in such a way that $ x \mapsto W_g(x) := \ed^{2 \pi \iu g(x)}$ is a smooth function that satisfies $\Delta_{\ep} = g^{-1} \big(  \Delta'_{\delta(\ep)}   \big)$ for all such $\ep$, and such that $W_g$ is invertible on $\Delta_{\ep}$. Then lemma \ref{lem:ac_spectral_mapping_2} applies to give
	$$ W_g \big( \sigma_{ac}(\hat H) \cap \Delta_{\ep} \big) = \sigma_{ac} \big( W_g(\hat H) \big) \cap \Delta'_{\delta} = \Delta'_{\delta} $$.
	where we used that the absolutely continuous spectrum of $W_g(\hat H)$ is the whole unit circle.

	Since $W_g$ is invertible on $\Delta_{\ep} = W_g^{-1}(\Delta'_\delta)$ it follows that $\Delta_{\ep}$ is contained in the absolutely continuous spectrum of $\hat H$ for any $\ep$. Since the absolutely continuous spectrum is a closed set, it follows that the gap $\Delta$ is contained in the absolutely continuous spectrum of $\hat H$, as required.
\end{proofof}

\appendix

\section{Spectral mapping for absolutely continuous spectrum}

Let $A$ be a bounded self-adjoint operator on a separable Hilbert space $\caH$ and Let $g : \R \rightarrow U(1)$ be a smooth function. By the spectral calculus this defines a unitary operator $B = g(A)$. We investigate the relation between the absolutely continuous spectra of $A$ and $B$.

Let $O$ be $A$ or $B$ and let $X$ be $\R$ if $O$ is self-adjoint, and $U(1)$ if $O$ is unitary. Take a continuous function $f \in C(X)$ from $X$ to $\C$ and a state $\psi \in \caH$. Through the continuous functional calculus we define a bounded linear functional on $C(X)$ by
$$ f \mapsto \langle \psi, f(O) \psi \rangle $$
By Riesz-Markov there is associated to this functional a unique measure $\mu^{(O)}_{\psi}$ on $X$ such that $\mu^{(O)}_{\psi}(X \setminus \sigma(O)) = 0$ and
$$ \langle \psi, f(O) \psi \rangle = \int_{X} \dd \mu^{(O)}_{\psi}(\lambda) \, f(\lambda). $$

Let $\caH^{(O)}_{ac} = \{ \psi \, | \, \mu^{(O)}_{\psi} \, \text{is absolutely continuous w.r.t. Lebesgue} \}$ and $P^{(O)}_{ac}$ the orthogonal projection onto this subspace. Then $O_{ac} = P^{(O)}_{ac} O P^{(O)}_{ac}$ is the absolutely continuous part of $O$ and $\sigma_{ac}(O) = \sigma(O_{ac})$.

We have that $x \in \sigma_{ac}(O)$ if and only if $x \in \supp \mu^{(O)}_{\psi}$ for some $\psi \in \caH^{(O)}_{ac}$.

\begin{lemma} \label{lem:ac_spectral_mapping_1}
	If $g : \R \rightarrow U(1)$ is continuous and invertible on $\sigma(A)$, then
	$$ \sigma_{ac}(g(A)) = g \big( \sigma_{ac}(A) \big). $$
\end{lemma}

\begin{proof}
	We have $x \in \sigma_{ac}(A)$ if and only if there is a $\psi \in \caH^{(A)}_{ac}$ such that $x \in \supp \mu^{(A)}_{\psi}$. To the same state $\psi$ is associated the linear functional
	$$ f \mapsto \langle \psi, f(g(A)) \psi \rangle $$
	and the unique measure $\mu^{(g(A))}_{\psi}$ on $U(1)$ such that
	$$ \langle \psi, f(g(A)) \psi \rangle = \int_{U(1)} \dd \mu^{(g(A))}_{\psi}(\lambda) f(\lambda)\;. $$
At the same time, we have
	$$  \langle \psi, f(g(A)) \psi \rangle = \int_{\R} \dd \mu^{(A)}_{\psi}(\lambda) f(g(\lambda)). $$
	The spectral measures $\mu^{(A)}_{\psi}$ and $\mu^{(g(A))}_{\psi}$ are therefore related by a change of variable:
	$$ \mu^{(g(A))}_{\psi}(\Delta) = \mu^{(A)}_{\psi}(g^{-1}(\Delta)) $$
	for any Borel set $\Delta$. Since $g$ is continuous and invertible, we have then that the measure $\mu^{(A)}_{\psi}$ is absolutely continuous w.r.t. Lebesgue on $\R$ if and only if $\mu^{(g(A))}$ is absolutely continuous w.r.t. Lebesgue on $U(1)$, and their supports are related by
	$$ \supp \mu^{(g(A))}_{\psi} = g \left(  \supp \mu^{(A)}_{\psi}  \right) $$
	\ie $x \in \supp \mu^{(A)}_{\psi}$ if and only if $g(x) \in \supp \mu^{(g(A))}_{\psi}$. From this we conclude the proof of the lemma.
\end{proof}

We now extend this result a bit as follows,

\begin{lemma} \label{lem:ac_spectral_mapping_2}
	If $g : \R \rightarrow U(1)$ is differentiable on $\sigma(A)$ and invertible on $\Delta_A = g^{-1}(\Delta_B)$ for some open subset $\Delta_B \subset U(1)$, then if $B = g(A)$ we have
	$$ g \big( \sigma_{ac}(A) \cap \Delta_A \big) = \sigma_{ac}(B) \cap \Delta_B $$
\end{lemma}

\begin{proof}
	Let $P^{(A)}_{\Delta_A}$ be the spectral projection for $A$ on $\Delta_A$ and $P^{(B)}_{\Delta_B}$ the spectral projection of $B$ on $\Delta_B$. Then letting $A' = P^{(A)}_{\Delta_A} A P^{(A)}_{\Delta_A}$ and $B' = P^{(B)}_{\Delta_B} B P^{(B)}_{\Delta_B}$ we have $g(A') = B'$ and $g : \Delta_A = \sigma(A') \rightarrow \Delta_B$ is differentiable and invertible. Therefore, lemma \ref{lem:ac_spectral_mapping_1} applies and gives
	$$ \sigma_{ac}(B') = g(\sigma_{ac}(A')). $$
	But $\sigma_{ac}(A') = \sigma_{ac}(A) \cap \Delta_A$ and $\sigma_{ac}(B') = \sigma_{ac}(B) \cap \Delta_B$ so
	$$ \sigma_{ac}(B) \cap \Delta_B = g \big( \sigma_{ac}(A) \cap \Delta_A \big) $$
	as required.

\end{proof}

\textbf{Acknowledgements}

A.B. and A.H.W. were supported by VILLIUM FONDEN through the QMATH Centre of Excellence (grant no. 10059). A.H.W. thanks the VILLIUM FONDEN for its support with a Villium Young Investigator Grant (grant no. 25452).

\bibliographystyle{plain}
\bibliography{QuantWalk}

\end{document}